\documentclass[11pt]{article}
\usepackage[utf8]{inputenc}

\usepackage[backend=biber]{biblatex}
\usepackage{
  fullpage, microtype, authblk,
  booktabs, diagbox, algorithm,
  amsfonts, amsmath, amssymb, amsthm, bm,
  url, tikz, subcaption, relsize,
  graphicx,
}
\usepackage[font=small,labelfont=bf]{caption}
\usepackage[inline]{enumitem}
\usepackage[normalem]{ulem}
\usepackage[noend]{algpseudocode}
\usepackage{nicefrac}
\usepackage{xcolor}
\usepackage[title]{appendix}
\usepackage[none]{hyphenat}

\MakeRobust{\Call}

\newcommand{\mund}{\mathunderscore}

\newcommand{\G}{\ensuremath{\mathcal{G}}}

\newcommand{\J}{\ensuremath{\mathcal{J}}}

\newcommand{\TTC}{{\text{\footnotesize TTC}}}
\newcommand{\TTCs}{{\text{\footnotesize TTC}s}}
\newcommand{\Rtuple}{{\footnotesize R}-tuple}
\newcommand{\Rtuples}{{\footnotesize R}-tuples}
\newcommand{\alga}{\texttt{{add\mund{}contact(u,v,t)}}}
\newcommand{\algb}{\texttt{can\mund{}reach($u$,$v$,$t_1$,$t_2$)}}
\newcommand{\algc}{\texttt{is\mund{}connected($t_1$,$t_2$)}}
\newcommand{\algd}{\texttt{reconstruct\mund{}journey($u$,$v$,$t_1$,$t_2$)}}
\newcommand{\alganame}{\texttt{add\mund{}contact}}
\newcommand{\algbname}{\texttt{can\mund{}reach}}
\newcommand{\algcname}{\texttt{is\mund{}connected}}
\newcommand{\algdname}{\texttt{reconstruct\mund{}journey}}
\newcommand{\Rleft}{\ensuremath{\mathcal{R}_{\text{\scriptsize left}}}}
\newcommand{\Rleftstar}{\ensuremath{\mathcal{R}^{*}_{\text{\scriptsize left}}}}
\newcommand{\Rleftuv}{\ensuremath{\mathcal{R}^{u,v}_{\text{\scriptsize left}}}}
\newcommand{\Rright}{\ensuremath{\mathcal{R}_{\text{\scriptsize right}}}}
\newcommand{\Rrightstar}{\ensuremath{\mathcal{R}^{*}_{\text{\scriptsize right}}}}
\newcommand{\Rrightuv}{\ensuremath{\mathcal{R}^{u,v}_{\text{\scriptsize right}}}}
\newcommand{\leftsubset}{\ensuremath{\subseteq_{\text{\scriptsize left}}}}
\newcommand{\rightsubset}{\ensuremath{\subseteq_{\text{\scriptsize right}}}}
\newcommand{\Btree}{B$^+$-tree}
\newcommand{\Btrees}{B$^+$-trees}

\newtheorem{definition}{Definition}

\newtheorem{theorem}{Theorem}

\newtheorem{lemma}[theorem]{Lemma}

\title{A Dynamic Data Structure for Representing Temporal Transitive Closures on Disk}

\author[1]{Luiz F. Afra Brito}
\author[1]{Marcelo Albertini}
\author[1]{Bruno A. N. Traven\c{c}olo}

\affil[1]{Federal University of Uberl\^{a}ndia, Brazil}

\begin{document}

\maketitle

\begin{abstract}
  Temporal graphs represent interactions between entities over time.
  These interactions may be direct, a contact between two vertices at some time instant, or indirect, through sequences of contacts called journeys.
  Deciding whether an entity can reach another through a journey is useful for various applications in complex networks.
  In this paper, we present a disk-based data structure that maintains temporal reachability information under the addition of new contacts in a non-chronological order.
  It represents the \emph{timed transitive closure} (\TTC) by a set of \emph{expanded} \Rtuples{} of the form $(u, v, t^-, t^+)$, which encodes the existence of journeys from vertex $u$ to vertex $v$ with departure at time $t^-$ and arrival at time $t^+$.
  Let $n$ be the number of vertices and $\tau$ be the number of timestamps in the lifetime of the temporal graph.
  Our data structure explicitly maintains this information in linear arrays using $O(n^2\tau)$ space so that sequential accesses on disk are prioritized.
  Furthermore, it adds a new unsorted contact $(u, v, t)$ accessing $O\left(\nicefrac{n^2\tau}{B}\right)$ sequential pages in the worst-case, where $B$ is the of pages on disk;
  it answers whether there is of a journey from a vertex $u$ to a vertex $v$ within a time interval $[t_1, t_2]$ accessing a single page;
  it answers whether all vertices can reach each other in $[t_1, t_2]$;
  and it reconstructs a valid journey that validates the reachability from a vertex $u$ to a vertex $v$ within $[t_1, t_1]$ accessing $O\left(\nicefrac{n\tau}{B}\right)$ pages.
  Our experiments show that our novel data structure are better that the best known approach for the majority of cases using synthetic and real world datasets.
\end{abstract}

\section{Introduction}

Temporal graphs represent interactions between entities over time.
These interactions often appear as contacts at specific timestamps.
Entities can also interact indirectly with each other by chaining several contacts.
For example, in a communication network, devices that are physically connected can send new messages or propagate received ones; thus, by first sending a new message and repeatedly propagating messages over time, remote entities can communicate indirectly.
Time-respecting paths in temporal graphs are known as temporal paths, or simply \textit{journeys}, and when a journey exists from one vertex to another one, we say that the first can \textit{reach} the second.

In a computational environment, it is often useful to check whether entities can reach each other~\cite{tang2010characterising,cacciari1996atemporal,whitbeck2012temporal,wu2017mining,williams2016spatio,bedogni2018temporal,martensen2017spatio,bryce2007atutorial}.
Beyond the sole reachability, some applications also require the ability to reconstruct a journey if one exists~\cite{wu2017mining,betsy2007spatio,zeng2014visualizing,hasan2011making,hurter2014bundled,vera2016querying}.
In standard graphs, the problem of updating reachability information is known as \emph{dynamic connectivity} and it has been extensively studied~\cite{tctech1,tctech5,2hopschema0,2hopschemadynamic3,onlinesearch0,onlinesearch1}.
In temporal graphs, fewer studies addressed this problem.
For instance, in~\cite{barjon2014testing,whitbeck2012temporal}, the authors assume that input is chronologically ordered and they give worst-case complexities.
In~\cite{wu2016reachability}, the authors assume non-chronological input, whereas their strategy is optimized for the average case.

Particularly to our interest, in~\cite{paper1}, the authors considered the \emph{only-incremental} problem, which supports only the addition of unsorted contacts.
Their data structure supports the following four operations, where, by convention, $\mathcal{G}$ is a temporal graph, $u$ and $v$ are vertices of $\mathcal{G}$, and $t, t_1$, and $t_2$ are timestamps:
\begin{enumerate}
  \item \alga, which updates information based on a contact from $u$ to $v$ at time $t$;
  \item \algb, which returns true if $u$ can reach $v$ within the interval $[t_1, t_2]$;
  \item \algc, which returns true if $\mathcal{G}$ restricted to the interval $[t_1,t_2]$ is temporally connected, \textit{i.e.}, all vertices can reach each other within the interval $[t_1, t_2]$; and
  \item \algd, which returns a journey (if one exists) from $u$ to $v$ occurring within the interval $[t_1, t_2]$.
\end{enumerate}

Their update algorithm maintains a \textit{timed transitive closure} (\TTC), a concept that generalizes the transitive closure for temporal graphs based on \textit{reachability tuples} (\Rtuples), in the form $(u, v, t^-, t^+)$, representing journeys from vertex $u$ to $v$ departing at $t^-$ and arriving at $t^+$.
Their data structure uses $O\left(n^2\tau\right)$ space while supporting \alganame, \algbname, \algcname, and \algdname, respectively, in $O\left(n^2\log\tau\right)$, $O\left(\log\tau\right)$, $O\left(n^2\log\tau\right)$, and $O\left(k\log\tau\right)$ worst-case time, where $n$ is the number of vertices of the temporal graph, $\tau$ is the number of time instances, and $k$ is the length of the resulting journey.

However, they keep their data structure in primary memory and the cost of storing and maintaining large \TTCs{} is prohibitive.
We conducted a simple experiment to show how much space is necessary  for temporal reachability.
First, we generated random temporal graphs using the Edge-Markovian Evolving Graph (EMEG) model~\cite{emeg}.
In this model, if an edge is active at time $t - 1$, then it has probability $p$ of disappearing at time $t$, otherwise, it has probability $q$ of appearing at time $t$.
We represented temporal graphs in memory using adjacency matrices storing, in each cell, timestamps at which edges are active.
Then, we built the corresponding \TTCs{} using the approach described in~\cite{paper1}.
In this experiment, we varied the number of vertices $n$ and the number of time instances $\tau$ while fixing $p = 0.1$ and $q = 0.3$.

In Table~\ref{tab:motivation2}, we see, for example, that a temporal graph with $512$ vertices and $\tau = 64$ produced by the EMEG model has $2.8$ million contacts, and we needed around $33$ MBs of space to store it in memory.
Besides, we needed around $156$ MBs of space to store the corresponding \TTC, which, in this case, it is almost five times the space needed to store the temporal graph.

\begin{table}
  \begin{center}
    \begin{tabular}{rrrrrrr}
      \toprule
      $n$ & $\tau$ & $|C|$ & $data(\G)$ & $data(\TTC)$ \\
      \midrule
      32  & 8   & 1268    & 0.02  & 0.01   \\
      32  & 16  & 2670    & 0.03  & 0.12   \\
      32  & 32  & 5249    & 0.06  & 0.24   \\
      64  & 8   & 5539    & 0.08  & 0.27   \\
      64  & 16  & 10908   & 0.14  & 0.54   \\
      64  & 32  & 21421   & 0.26  & 1.08   \\
      64  & 64  & 42671   & 0.50  & 2.17   \\
      128 & 8   & 21203   & 0.31  & 1.13   \\
      128 & 16  & 43011   & 0.55  & 2.30   \\
      128 & 32  & 86746   & 1.06  & 4.63   \\
      128 & 64  & 173479  & 2.05  & 9.31   \\
      256 & 8   & 86574   & 1.24  & 4.67   \\
      256 & 16  & 174970  & 2.25  & 9.51   \\
      256 & 32  & 346994  & 4.22  & 19.17  \\
      256 & 64  & 696436  & 8.22  & 38.54  \\
      512 & 8   & 349114  & 5.00  & 19.01  \\
      512 & 16  & 702294  & 9.04  & 38.66  \\
      512 & 32  & 1396033 & 16.98 & 78.00  \\
      512 & 64  & 2800520 & 33.05 & 156.64 \\
      \bottomrule
    \end{tabular}
  \end{center}
  \caption[Space for storing temporal graphs and their corresponding \TTCs]{Space for storing temporal graphs with $n$ vertices, $\tau$ time instances and $|C|$ contacts, and their corresponding \TTCs. Columns $data(\G)$ and $data(\TTC)$ represent, respectively, the space in megabytes of the generated temporal graphs and their \TTCs.}\label{tab:motivation2}
\end{table}

Next, we built a linear regression model with the data presented in Table~\ref{tab:motivation2} in order to extrapolate the input parameters.
Consider, for example, the scenario in which one million people use a bluetooth device that registers when and who gets close to each other and sends this information to a centralized server.
Consider also that each individual makes in average $30$ contacts per day.
In this setting, by using our model, we could check that a centralized server would require at least $100$ GBs of space in less than a year to store just the plain contacts as a temporal graph.
If one needs to support reachability queries by using a \TTC{}, it would be necessary roughly $600$ GBs of space.

Motivated by such scenarios, we investigate the problem of maintaining \TTCs{} on disk.
A simple, but not efficient, approach would be to naively implement a data structure for disk based on the approach described in~\cite{paper1}.
Briefly, their strategy maintains self-balanced binary search trees (BSTs) containing time intervals for each
pair of vertices in order to retrieve reachability information.
However, this approach does not consider data locality, thus each update operation would randomly access an excessive amount of pages on disk to retrieve information from each BST.
For instance, if we use \Btrees~\cite{bplustree} as a replacement for BSTs, their algorithm for answering \alganame{} would access $O\left(n^2\right)$ \Btrees{} and, in each \Btree, it would access $O\left(\log_B{\tau}\right)$ pages, where $B$ is the page size, resulting in $O\left(n^2\log_B{\tau}\right)$ random accesses on disk.
Therefore, we need a novel approach that better organizes data on disk.

We propose in this paper an incremental disk-based data structure that reduces the number of disk accesses for both update and query operations while prioritizing sequential accesses.
The core idea of our novel approach is to maintain explicitly an \textit{expanded} set of non-redundant \Rtuples{} containing $n^2\tau$ elements.
Conceptually, we maintain it using two $3$-dimensional arrays, $M_{out}$ and $M_{in}$, of size $n \times \tau \times n$ such that $M_{out}[u, t^-, v] = t^+$ and $M_{in}[v, t^+, u] = t^-$.
The former supports querying the earliest arrival time $t^+$ a journey departing from vertex $u$ at time $t^-$ can arrive at vertex $v$, and the latter supports querying the latest departure time $t^-$ a journey arriving to vertex $v$ at time $t^+$ can depart from vertex $u$.

Our algorithm to compute \alganame{} eagerly updates both arrays accessing $O\left(\nicefrac{n^2\tau}{B}\right)$ disk pages in the worst case.
Despite having a linear factor on $\tau$ instead of logarithmic, the expected cost of our update routine reduces considerably as we insert new contacts.
This is because journey schedules become stricter and the probability of replacing them with faster ones reduces.
Since we explicitly maintain reachability information, our algorithms to answer \algbname{}, \algcname{}, and \algdname{} access, respectively, one, $\Theta\left(\nicefrac{n^2}{B}\right)$ and $\Theta\left(\nicefrac{n}{B}\right)$ pages.

We compare our novel data structure with a na\"{i}ve adaptation of the approach introduced in~\cite{paper1} using \Btrees{} as replacement for BSTs.
Our experiments show that our novel data structure performs better on the synthetic datasets and on the majority of real-world datasets we used.
Even though the worst-case complexity of our algorithm for the \alga{u}{v}{t} operation is linear in $\tau$ instead of logarithmic, it runs much faster on average.
We attribute this behavior to the fact that as new contacts are inserted, our data structure updates on average only a few cells of both arrays $M_{out}$ and $M_{in}$.

We organized this paper as follows.
In Section~\ref{sec:definitions}, we present the definitions used throughout this paper.
In Section~\ref{sec:disk-timed-transitive-closure}, we define our expanded set of \Rtuples, introduce our new data structure to represent \TTCs{} on disk, and provide low-level primitives for manipulating them.
In Section~\ref{sec:disk-operations}, we describe our algorithms for each operation using our data structure along with their complexities in terms of number of disk accesses.
In Section~\ref{sec:disk-experiments}, we investigate the execution of our algorithms by comparing them with our implementation using \Btrees.
Finally, Section~\ref{sec:disk-conclusions} concludes with some remarks and open questions.

\section{Definitions}\label{sec:definitions}

Following the definition in~\cite{casteigts2012time}, a temporal graph is a tuple $\mathcal{G} = (V, E, \mathcal{T}, \rho, \zeta)$.
Sets $V$ and $E \subseteq{} V \times V$ represent the vertices and the edges of the underlying standard graph.
Interval $\mathcal{T} = [1, \tau] \subset \mathbb{N}$ describes the lifetime of the temporal graph.
We consider in this paper that $E$ is a set of directed edges.
Functions $\rho: E \times \mathcal{T} \to\{0, 1\}$ and $\zeta: E \times\mathcal{T} \mapsto \mathbb{N}$ are, respectively, the \emph{presence function} and the \textit{latency function}.
The presence function expresses whether an edge is present at a time instant.
We also call $(u, v, t)$ a \emph{contact} in ${\mathcal{G}}$ if $\rho((u, v), t) = 1$.
The latency function expresses the duration of an interaction for an edge at a time.
Here, we use a constant latency function, namely $\zeta = \delta$, where $\delta$ is any fixed positive integer.

We define reachability in temporal graphs in a time-respecting way, by requiring that a path travels along non-decreasing ($\delta = 0$) or increasing ($\delta \geq 1$) times.
These paths are called temporal paths or journeys interchangeably.

\begin{definition}[Journey]\label{def:journey}
  A journey from $u$ to $v$ in $\G$ is a sequence of contacts $\mathcal{J} = \langle c_1, c_2, \ldots, c_k \rangle$, whose sequence of underlying edges form a valid time-respecting path from $u$ to $v$.
  For each contact $c_i = (u_i, v_i, t_i)$, it holds that $\rho((u_i, v_i), t_i) = 1$, $v_i = u_{i + 1}$, and $t_{i+1} \ge t_i + \delta$ for $i \in [1, k-1]$.
  We say that $departure (\mathcal{J}) = t_1$, $arrival (\mathcal{J}) = t_{k} + \delta$ and $duration (\mathcal{J}) = arrival (\mathcal{J})- departure (\mathcal{J})$.
  A journey is \emph{trivial} if it comprises a single contact.
\end{definition}

\begin{definition}[Reachability]%
    \label{def:reachability}
    A vertex $u$ can \emph{reach} a vertex $v$ within time interval $[t_1, t_2]$ iff there is a journey $\J$ from $u$ to $v$ in $\mathcal{G}$ that departs at $departure(\J) \ge t_1$ and arrives at $arrival(\J) \le t_2$.
\end{definition}

Just as the number of paths in a standard graph, the number of journeys in a temporal graph could be too large to be stored explicitly (typically, factorial in $n$).
To avoid this problem, \Rtuples{} capture the fact that a vertex can reach another one within a certain time interval without storing the corresponding journeys~\cite{paper1}.

\begin{definition}[\Rtuple] A \Rtuple{} is a tuple $r=(u, v, t^-, t^+)$, where $u$ and $v$ are vertices in $\mathcal{G}$, and $t^-$ and $t^+$ are timestamps in $\mathcal{T}$. It encodes the fact that vertex $u$ can reach vertex $v$ through a journey $\mathcal{J}$ such that $departure (\mathcal{J}) = t^-$ and $arrival (\mathcal{J}) = t^+$.
  If several such journeys exist, then they are all represented by the same \Rtuple.
\end{definition}



Lastly, given a temporal graph $\mathcal{G}$, the timed transitive closure (\TTC) of $\mathcal{G}$ is a directed multigraph on the same set of vertices, whose edges correspond to the minimal, \textit{i.e.}, non-redundant, set of \Rtuples{} of $\mathcal{G}$.
The purpose of \TTCs{} is to encode reachability information among vertices, parametrized by time intervals, so that one can subsetuently decide if a new contact can be composed with existing journeys.
In paper~\cite{paper1}, the authors showed that there are $O\left(n^2\tau\right)$ non-redundant \Rtuples{} in a temporal graph $\mathcal{G}$ and it comprises those whose intervals do not include each other for the same pair of vertices, \textit{i.e.}, only information regarding the fastest journeys.

\section{Disk-Based Timed Transitive Closure}\label{sec:disk-timed-transitive-closure}

In this section, we describe our novel approach to maintain \TTCs{} in secondary memory.
First, in Section~\ref{ssec:R-tuples}, we define the concept of an \emph{expanded} set of representative \Rtuples{} and show that it has size $\Theta\left(n^2\tau\right)$.
Then, in Section~\ref{ssec:disk-datastructure}, we introduce our new data structure that uses this expanded set in order to improve the maintenance of data in non-uniform access storages and provide direct access to reachability information.

\subsection{Expanded Reachability Tuples (Expanded \Rtuples)}\label{ssec:R-tuples}

The data structure introduced in~\cite{paper1} spreads the minimal set of \Rtuples{} into multiple BSTs, each one concerning a unique pair of vertices.
The authors store these BSTs in separated regions of memory and, therefore, the organization of data is not optimal when working with storages that have non-uniform access time.

In order to mitigate this problem, we define an expanded set of \Rtuples{} $(u, v, t^-, t^+)$ that is easier to maintain sequentially, since we can use continuous arrays indexed by $t^-$ or $t^+$.
First, we define the \emph{left} and \emph{right} expansion of a single \Rtuple.

\begin{definition}[Left and right expansion]
  The left expansion of a \Rtuple{} $r = (u, v, t^-, t^+)$ is the set containing all \Rtuples{} $(u, v, t, t^+)$ for $1 \leq t \leq t^-$.
  Similarly, the right expansion of $r$ is the set containing all \Rtuples{} $(u, v, t^-, t)$ for $t^+ \leq t \leq \tau + \delta$.
\end{definition}

The \Rtuples{} produced by the left expansion of a \Rtuple{} $r$ are valid because a source vertex departing earlier can simply wait until the departure time of $r$, and take the original journey described by $r$.
Similarly, the \Rtuples{} produced by the right expansion of $r$ are valid because, after taking the original journey described by $r$, a destination vertex can simply wait until the arrival time of the new \Rtuple.

Applying both expansions to each \Rtuple{} in a set $\mathcal{R}$ and taking the union of the sets produced by the same expansion creates two separated expanded sets, the \emph{left-expanded} set \Rleft, and the \emph{right-expanded} set \Rright.
For each expanded set, we define an inclusion operator.

\begin{definition}[Left and right inclusion]
  Given any two \Rtuples{} $r_1 = (u_1, v_1, t^-_1, t^+_1)$ and $r_2 = (u_2, v_2, t^-_2, t^+_2)$ in \Rleft, $r_1 \leftsubset r_2$ if and only if $u_1 = u_2$, $v_1 = v_2$, $t^-_1 = t^-_2$, and $t^+_1 \leq t^+_2$.
  Similarly, if $r_1$ and $r_2$ are in \Rright, $r_1 \rightsubset r_2$ if and only if $u_1 = u_2$, $v_1 = v_2$, $t^-_1 \geq t^-_2$, and $t^+_1 = t^+_2$.
\end{definition}

However, \Rtuples{} produced by expansion can share redundant information.
For example, consider the \Rtuples{} $r_1 = (a, b, 2, 7)$ and $r_2 = (a, b, 2, 9)$.
Both \Rtuples{} represent journeys that departs from vertex $a$ at time $2$ and arrives at vertex $b$, one at time $7$ and the other at time $9$.
In this case, $r_2$ can be safely discarded since we can take a journey represented by $r_1$ ending at time $7$ and wait at vertex $v$ until time $9$.
Redundancy of \Rtuples{} in \Rleft{} and \Rright{} is treated differently using their corresponding inclusion operators.

\begin{definition}[Left and right redundancy]\label{def:redundancy}
  Let $r \in \Rleft$, $r$ is called left-redundant in \Rleft{} if there is $r' \in \Rleft$ such that $r' \leftsubset r$.
  Similarly, if $r \in \Rright$, $r$ is called right-redundant in \Rright{} if there is $r' \in \Rright$ such that $r' \rightsubset r$.
  A set \Rleftstar{} with no left-redundant \Rtuple{} is called left non-redundant and a set \Rrightstar{} with no right-redundant \Rtuple{} is called right no-redundant.
\end{definition}

\begin{lemma}\label{lemma:number}
  The maximum size of a left non-redundant or right non-redundant set of \Rtuples{} for \G{} is $\Theta\left(n^2\tau\right)$.
\end{lemma}

\begin{proof}
  It suffices to prove that the maximum number of pairwise incomparable \Rtuples{} in \Rleft{} and \Rright{} is $O\left(n^2\tau\right)$, since some graphs induce $\Theta\left(n^2\tau\right)$ incomparable \Rtuples{} from unexpanded sets, see~\cite{paper1}.
  We prove that the sizes of \Rleft{} and \Rright{} are $O\left(n^2\tau\right)$ as follows.
  There are $\Theta\left(n^2\right)$ ordered pairs of vertices.
  Thus, it is enough to show that for each pair $(u, v)$, the number of incomparable \Rtuples{} in \Rleftuv{} and \Rrightuv{}, whose source vertex is $u$ and destination vertex is $v$, is $\Theta\left(\tau\right)$.
  Let \Rleftuv{} be a left non-redundant set of such \Rtuples{}, as every incomparable \Rtuple{} has different arrival timestamps, $|\Rleft| \leq \tau$.
  Similarly, let \Rrightuv{} be a right non-redundant set of such \Rtuples, as every incomparable \Rtuple{} has different departure timestamps, $|\Rright| \leq \tau$. \medskip\\
\end{proof}

\subsection{Encoding the \TTC{} on Disk}\label{ssec:disk-datastructure}

\begin{figure}
  \centering
  \begin{tabular}{cccc}
    \raisebox{0.5\height}{\includegraphics[page=4, width=0.3\textwidth]{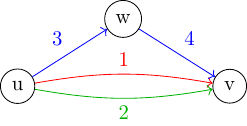}} &&
    \includegraphics[page=12, width=0.2\textwidth]{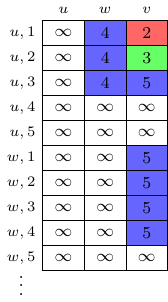} &
    \includegraphics[page=16, width=0.2\textwidth]{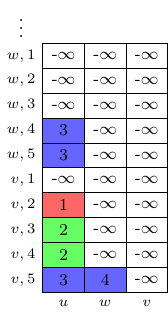} \\
    (a) $\mathcal{G}$ && (b) $M_{out}$  & (c) $M_{in}$  \\
  \end{tabular}
  \caption{
    Temporal graph and its associated reachability data structure. In~(a), we show a temporal graph with three vertices. Numbers on edges represent the time in which edges are active. Edges with the same color form a journey from vertex $u$ to vertex $v$. In~(b), we show the corresponding arrays $M_{out}$ and $M_{in}$ considering $\delta = 1$. Both arrays are depicted as 2-dimensional arrays by grouping their first two dimensions. For instance, $M_{out}[u, 2, w] = M_{out}[(u, 2), w] = 3$. Cells have the same color as the contacts, \textit{i.e.}, the edge at a timestamp, that originated the update. $M_{out}$ stores the minimum possible arrival timestamps to destinations and $M_{in}$ sotores the maximum possible departure timestamps from origins.
  }\label{fig:datastructure-disk-illustriation}
\end{figure}

We encode the \TTC{} using two $3$-dimensional arrays, $M_{out}[u, t^-, v] = t^+$ and $M_{in}[v, t^+, u] = t^-$, both with dimensions $n \times \tau \times n$, representing expanded sets of \Rtuples.
Each cell in $M_{out}$ represents a \Rtuple{} in $\Rleftstar$ by storing the earliest possible arrival time $t^+$ at which a vertex $u$ departing at time $t^-$ can reach a vertex $v$ through a journey.
If there is a cell $M_{out}[u, t^-, v] = t^+$, then all cells $M_{out}[u, t, v]$, for $t \in [1, t^- - 1]$ must have an arrival time $t_{\text{\scriptsize left}} \leq t^+$, since a journey from $u$ departing at a time $t < t^-$ can simply wait at vertex $u$ until time $t^-$ and then use the remaining path already described by $M_{out}[u, t^-, v] = t^+$.
Similarly, each cell in $M_{in}$ represents a \Rtuple{} in $\Rrightstar$ by storing the latest possible departure time $t^-$ at which a vertex $v$ can arrive at time $t^+$ to a vertex $u$ through a journey.
If there is a cell $M_{in}[v, t^+, u] = t^-$, then all cells $M_{in}[v, t, u]$, for $t \in [t^+ + 1, \tau + \delta]$, must have a departure time $t_{\text{\scriptsize right}} \geq t^+$, since a journey to $v$ arriving at a time $t > t^+$ can use the path already described by $M_{in}[v, t^+, u] = t^-$ and then simply wait at vertex $v$ until time $t$.
During the creation of a \TTC{}, $M_{out}$ cells are initialized with $\infty$
and $M_{in}$ cells with $-\infty$.
Figure~\ref{fig:datastructure-disk-illustriation} illustrates both $M_{out}$ and $M_{in}$.

Internally, we represent $M_{out}$ and $M_{in}$ as one-dimensional arrays using, respectively, the mapping functions $F_{out}\colon (u, t^-, v) \mapsto n (u \tau + \tau - (t^- + 1)) + v + 1$ and $F_{in}\colon (v, t^+, u) \mapsto n (v \tau + t^+ - \delta) + u + 1$.
Observing Figure~\ref{fig:datastructure-disk-illustriation}, $F_{out}$ arranges the cells of $M_{out}$ by row (left to right) and, for each source vertex, later departures come first.
$F_{in}$ also arranges $M_{in}$ by row but, in contrast, for each destination vertex, earlier arrivals come first.
By subtracting $\delta$ from $t^+$ in $F_{in}$, we ensure all $t^+$ values fit in $M_{in}$.
Thus, reading sequentially the range $[F_{out}(u, t^-, 1), F_{out}(u, t^-, n)]$ from $M_{out}$ gives direct access to the earliest arrival times to reach all vertices when departing from $u$ at time $t^-$.
Similarly, reading sequentially the range $[F_{in}(v, t^+, 1), F_{in}(v, t^+, n)]$ from $M_{in}$ gives direct access to the latest departure times to leave all vertices when arriving at $v$ at time $t^+$.

Finally, assuming a general function $F$ that maps to $F_{out}$, whether accessing $M_{out}$, or $F_{in}$, whether accessing $M_{in}$, we provide the following low-level operations for manipulating our data structures on disk:
\begin{enumerate}
  \item \Call{read\mund{}cell}{$M, w_1, t, w_2$}, which returns the value of $M$ ($M_{out}$ or $M_{in}$) at position $F(w_1, t, w_2)$;
  \item \Call{write\mund{}cell}{$M, w_1, t_1, w_2, t_2$}, which replaces the value of $M$ at position $F(w_1, t_1, w_2)$ with $t_2$;
  \item \Call{read\mund{}adjacency}{$M, w, t$}, which returns a list containing the values of $M$ in the interval $[F(w, t, 1), F(w, t, n)]$, \textit{i.e.}, the minimum possible timestamps to arrive at any vertex while departing from $w$ at timestamp $t$;
  \item \Call{write\mund{}adjacency}{$M, w, t, L$}, which replaces the values of $M$ values in the interval $[F(w, t, 1), F(w, t, n)]$ with the values of the list $L$, \textit{i.e.}, the maximum possible timestamps to depart from any vertex while arriving at $w$ at timestamp $t$.
\end{enumerate}

Operations (1) and (2) access $O\left(1\right)$ pages on disk, while operations (3) and (4) access  $O\left(\nicefrac{n}{B}\right)$ pages, where $B$ is the page size.

\section{\TTC{} Operations}\label{sec:disk-operations}

n this section, we describe algorithms for the operations described in~\cite{paper1}: the update operation \alga{}; the query operations \algb{} and \algc{}; and the reconstruction operation \algd.
In Section~\ref{subsec:alg1}, we present our algorithm for \alga{} that receives a contact and adds to our data structure  the reachability information related to the new available journeys passing thought it.
In Section~\ref{subsec:alg23}, we breafly describe algorithms for \algb{} and \algc{} since, as reachability information can be directly accessed, they are straightforward.
Finally, in Section~\ref{subsec:alg4}, we detail our algorithm for \algd{} that reconstructs a valid journey by concatenating one contact at a time.

\subsection{Update operation}\label{subsec:alg1}

An algorithm to perform \alga{} must first add the reachability information regarding the new trivial journey $\J_{triv}$ from vertex $u$ to vertex $v$ departing at time $t$ and arriving at time $t + \delta$.
Next, for all vertices $w^+$ that $v$ can reach when departing at a time later than or exactly $t + \delta$, the algorithm updates the reachability information from $u$ to $w^+$ whether the new available journey passing through $\J_{triv}$ has earlier arrival time.
Then, for all vertices $w^-$ that can reach $u$ when arriving at a time earlier than or exactly $t$, the algorithm updates the reachability information from $w^-$ to $v$ whether the new available journey passing through $\J_{triv}$ has later departure time.
Finally, the algorithm must consider all new available journeys from vertices $w^-$ to vertices $w^+$ that pass through $\J_{triv}$ and update the current reachability information if necessary.

Algorithm~\ref{alg:1-disk} describes the maintenance of both arrays $M_{out}$ and $M_{in}$ when inserting a new contact.
In line 1, the algorithm checks if the structure already has the information of the new contact $(u, v, t)$.
If it still has not, in line 2, it retrieves the latest departure timestamps of journeys departing from vertices $w^-$ and arriving at vertex $u$ at timestamp $t$ as an array $T^-$.
In line 3, the algorithm retrieves the earliest arrival timestamps of journeys departing from vertex $v$ at timestamp $t + \delta$ and arriving at vertices $w^+$ as an array $T^+$.
In lines 4 and 5, it sets the reachability information about the new trivial journey $\J_{triv} = (u, v, t)$ that departs at timestamp $t$ and arrives at $t + \delta$.
From line 6 to 14, the algorithm eagerly updates all cells $M_{out}[w^-, t', w^+] = t^+$ for $t^- \geq t' \geq 1$.
In this part, the algorithm proceeds by first iterating through all vertices $w^-$, \textit{i.e.}, those that reached $u$ before than or exactly at timestamp $t$, and retrieving their departure timestamps $t^-$.
Then, it progressively retrieves the current arrival timestamps to reach vertices $w^+$ when departing at timestamp $t'$, by reading the range $[F_{out}(w^-, t', 1), F_{out}(w^-, t', n)]$, and updates it whether the new journeys passing through $\J_{triv}$ have earlier arrival timestamps.
Note that, vertices that could not reach $u$ before than or exactly at timestamp $t$ have their arrival time equals to $-\infty$; therefore, they are not considered in the while loop starting at line 8.
This process continues until the current reachability information in the whole range does not change or $t' < 1$.
Similarly, from line 6 to 14, the algorithm eagerly updates all cells $M_{in}[w^+, t'', w^-] = t^-$ for $t^+ \leq t'' \leq \tau + \delta$.
The algorithm proceeds by first iterating through all vertices $w^+$, \textit{i.e.}, those that $v$ can reach departing after or exactly at timestamp $t + \delta$, and retrieving their arrival timestamps $t^+$.
Then, it progressively retrieves the current departure timestamps in which vertices $w^-$ departs present in range $[F_{in}(w^+, t'', 1), F_{in}(w^+, t'', n)]$ and updates it whether the new available journeys passing through $\J_{triv}$ have later departure timestamps.
This process continues until the current reachability information in the range does not change or $t'' > \tau + \delta$.
Figure~\ref{fig:datastructure-disk-example} illustrates the addition of new contacts to a temporal graph along with the maintenance of the arrays $M_{out}$ and $M_{in}$.

\begin{algorithm*}
  \caption{\alga}\label{alg:1-disk}
  \begin{algorithmic}[1]
    \Require{
      $u,v \in V$ with $u \neq{} v$,
      $n = |V|$,
      $t \in \mathcal{T}$,
      $\tau$,
      $\delta$,
      $M_{out}$,
      $M_{in}$
    }
    \If{$\Call{read\mund{}cell}{M_{out}, u, t, v} \neq t + \delta$}
    \Comment{check whether $(u, v, t)$ was inserted}
      \State{$T^- \gets \Call{read\mund{}adjacency}{M_{in}, u, t}$}
      \State{$T^+ \gets \Call{read\mund{}adjacency}{M_{out}, v, t + \delta}$}
      \State{$T^-[u] \gets t$}
      \Comment{add the new trivial journey information}
      \State{$T^+[v] \gets t + \delta$}
      \For{$w^- \text{ from } 1 \text{ up to } n$}
      \Comment{will update $M_{out}$ with new journeys from $w^-$}
        \State{$t' \gets T^-[w^-]$}
        \While{$t' \neq -\infty$ and $t' \geq 1$}
        \Comment{loop for $t^- \geq t' \geq 1$}
          \State{$T^+_{cur} \gets \Call{read\mund{}adjacency}{M_{out}, w^-, t'}$}
          \State{$T^+_{cur}[w^+] \gets \min(T^+_{cur}[w^+], T^+[w^+])$ for $w^+ \in [1, n]$}
          \If{$T^+_{cur}$ has not changed}
            \State{\textbf{break}}
          \EndIf{}
          \State{$\Call{write\mund{}adjacency}{M_{out}, w^-, t', T^+_{cur}}$}
          \State{$t' \gets t' - 1$}
        \EndWhile{}
      \EndFor{}

      \For{$w^+ \text{ from } 1 \text{ up to } n$}
      \Comment{will update $M_{in}$ with new journeys to $w^+$}
        \State{$t'' \gets T^+[w^+]$}
        \While{$t'' \neq \infty$ and $t'' \leq \tau + \delta$}
        \Comment{loop for $t^+ \leq t'' \leq \tau + \delta$}
          \State{$T^-_{cur} \gets \Call{read\mund{}adjacency}{M_{in}, w^+, t''}$}
          \State{$T^-_{cur}[w^-] \gets \max(T^-_{cur}[w^-], T^-[w^-])$ for $w^- \in [1, n]$}
          \If{$T^-_{cur}$ has not changed}
            \State{\textbf{break}}
          \EndIf{}
          \State{$\Call{write\mund{}adjacency}{M_{in}, w^+, t'', T^-_{cur}}$}
          \State{$t'' \gets t'' + 1$}
        \EndWhile{}
      \EndFor{}
    \EndIf{}
  \end{algorithmic}
\end{algorithm*}

\begin{figure}
  \centering
  \begin{tabular}{cccc}
    \raisebox{0.5\height}{\includegraphics[page=2, width=0.3\textwidth]{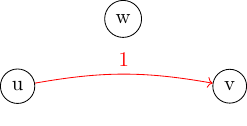}} &&
    \includegraphics[page=10, width=0.2\textwidth]{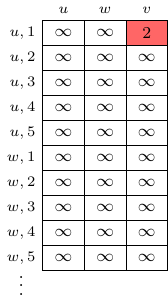} &
    \includegraphics[page=14, width=0.2\textwidth]{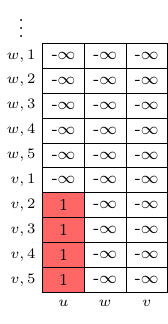} \\
    \multicolumn{4}{c}{(a) added contact $(u, v, 1)$} \\
  \end{tabular}
  \begin{tabular}{cccc}
    \raisebox{0.5\height}{\includegraphics[page=3, width=0.3\textwidth]{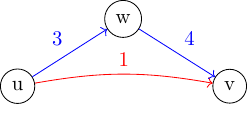}} &&
    \includegraphics[page=11, width=0.2\textwidth]{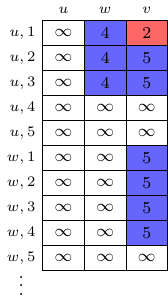} &
    \includegraphics[page=15, width=0.2\textwidth]{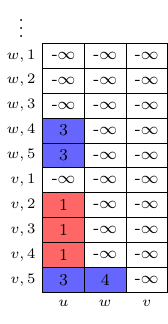} \\
    \multicolumn{4}{c}{(b) added contacts $(u, w, 3)$ and $(w, v, 4)$} \\
  \end{tabular}
  \begin{tabular}{cccc}
    \raisebox{0.5\height}{\includegraphics[page=4, width=0.3\textwidth]{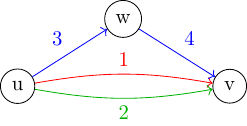}} &&
    \includegraphics[page=12, width=0.2\textwidth]{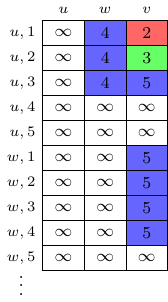} &
    \includegraphics[page=16, width=0.2\textwidth]{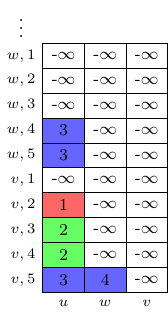} \\
    \multicolumn{4}{c}{(c) added contact $(u, v, 2)$} \\
  \end{tabular}
  \caption{
    Maintenance of our disk-based \TTC{}, encoded as two arrays, $M_{out}$ (left table) and $M_{in}$ (right table), in different scenarios.
    Both arrays are depicted as 2-dimensional arrays by grouping their first two dimensions.
    In~(a), the contact $(u, v, 1)$ is inserted in the temporal graph and thus $M_{out}$ and $M_{in}$ are updated using the information present in the left and right expansions of the \Rtuple{} $(u, v, 1, 2)$.
    In~(b), both contacts $(u, w, 3)$ and $(w, v, 4)$ are inserted and, additionally, a non-trivial journey from $u$ to $v$ becomes possible.
    Finally, in~(c), the contact $(u, v, 2)$ is inserted, allowing a faster journey departing from $u$ at time $2$ and triggering the update of some cells of $M_{out}$ and $M_{in}$.
  }\label{fig:datastructure-disk-example}
\end{figure}

\begin{theorem}\label{the:cost-algb-disk}
  Algorithm~\ref{alg:1-disk} access $O\left(\nicefrac{n^2\tau}{B}\right)$ pages on disk.
\end{theorem}

\begin{proof}
  The \textproc{read\mund{}cell} operation in line~1 access a single page.
  The two \textproc{read\mund{}adjacency} operations in lines~2 and~3 access  $O\left(\nicefrac{n}{B}\right)$ sequential pages each.
  In lines 4 and 5, the algorithm writes the reachability of the new trivial journey $J_{triv} = \{(u, v, t)\}$ in main memory.
  The for loop starting at line 6 iterates over $n$ vertices $w^-$ and, the while loop starting at line 8 iterates through $O(\tau)$ timestamps $t'$.
  At each of the $O(n\tau)$ iterations, it calls \textproc{read\mund{}adjacency} in order to read $n$ cells, and then (possibly) calls \textproc{write\mund{}adjacency} to write the $n$ cells back while accessing, in each operation, $O\left(\nicefrac{n}{B}\right)$ sequential pages.
  Due to our mapping function $F_{out}$, at every timestamp $t'$, the algorithm will read a page that is arranged sequentially on disk.
  The loop from line 15 to 23 does a similar computation.
\end{proof}

\subsection{Reachability and Connectivity Queries}\label{subsec:alg23}

Both algorithms for \algb{} and \algc{} are straightforward.
The algorithm to perform \algb{} comprises testing whether $\Call{read\mund{}cell}{M_{out}, u, t_1, v} \leq{} t_2$ while accessing only a single page from disk.
The algorithm to perform \algc{}, for each origin vertex $u \in V$, calls $tmp \gets \Call{read\mund{}adjacency}{M_{out}, u, t_1}$ and then for each destination vertex $v \in V$, it checks whether $tmp[v] \leq{} t_2$.
As soon as a check is negative, the answer is false; otherwise, it is true.
Therefore, while the algorithm reads all cells, it accesses $O\left(\nicefrac{n^2}{B}\right)$ pages on disk.

\subsection{Journey Reconstruction}\label{subsec:alg4}

For the \algd{} query, we need to augment each cell of $M_{in}$ with the first successor vertex of the corresponding journeys.
Algorithm~\ref{alg:1-disk} can be trivially modified to include this information.
For instance, the successor vertex of a trivial journey from a contact $(u, v, t)$ is the vertex $v$ since it is the first successor of $u$.
Thus, while composing previous \Rtuples{} in our update algorithm, one would need to read previous reachability information and compose than appropriately considering also the successor vertex present in each cell of $M_{in}$.

Algorithm~\ref{alg:2-disk} gives the details to process the \algd{} query.
Its goal algorithm is to reconstruct a journey by unfolding the intervals and successor fields.
In line 1, it initializes an empty journey $\J$.
In line 2, it retrieves the earliest timestamp $t^+$ a journey from vertex $u$ departing at timestamp $t_1$ can arrive at vertex $v$ by reading on disk the entry $M_{out}[u, t_1, v]$.
If $t^+ \leq{} t_2$, it starts reconstructing the resulting journey, otherwise, it returns an empty journey since there is no journey completely in the interval $[t_1, t_2]$.
From lines 4 to 10, it reconstructs the resulting journey by:
first, in lines~4 and~5, initializing the successor vertex $succ$ to $u$, and accessing on disk all the entries $M_{in}[v, t^+, w]$ for $w \in V$;
then, from lines~6 to~10, the joursey s reconstructed by iteratively accessing the next earliest departing timestamp $t^-$ and the corresponding successor vertex $next\mund{}succ$ that reaches $v$ at timestamp $t^+$ while concatenating the contact $(succ, next\mund{}succ, t^-)$ at the end of $\mathcal{J}$ and updating the current successor vertex.

\begin{algorithm*}
  \caption{\algd}\label{alg:2-disk}
  \begin{algorithmic}[1]
    \Require{$[t_1, t_2] \subset{\mathcal{T}}, u,v \in V$ with $u \neq v$}
    \State{$\mathcal{J} \gets \{\}$}

    \State{$(t^+, \_) \gets \Call{read\mund{}cell}{M_{out}, u, t_1, v}$}
    \If{$t^+ \leq{} t_2$}
      \State{$succ \gets u$}
      \State{$in \gets \Call{read\mund{}adjacency}{M_{in}, v, t^+}$}
      \While{$succ \neq v$}
        \State{$t^- \gets in[succ].t$}
        \State{$next\mund{}succ \gets in[succ].succ$}
        \State{$\mathcal{J} \gets \mathcal{J} \cup (succ, next\mund{}succ, t^-)$}
        \State{$succ \gets next\mund{}succ$}
      \EndWhile{}
    \EndIf{}
    \State{\textbf{return} $\mathcal{J}$}
  \end{algorithmic}
\end{algorithm*}

\begin{theorem}
  Algorithm~\ref{alg:2-disk} sequentially accesses $O\left(\nicefrac{n}{B}\right)$ pages on disk, where $n$ is the number of vertices and $B$ is the page size.
\end{theorem}

\begin{proof}
  The algorithm accesses one page by calling $\Call{read\mund{}cell}{M_{out}, u, t_1, v}$ in line~2.
  After that, it is known whether a journey exists or not.
  If a journey exists, it sequentially accesses $\nicefrac{n}{B}$ pages by calling $\Call{read\mund{}adjacency}{M_{in}, v, t^+}$ in line~5.
  The result of this call has all information needed to reconstruct a valid journey.
  Finally, in the loop from line~6 to line~8, the algorithm extends the resulting journey by one contact at each iteration using information already in memory.
  Thus, the number of pages accessed is dominated by the call $\Call{read\mund{}adjacency}{M_{in}, v, t^+}$.
\end{proof}

\section{Experiments}\label{sec:disk-experiments}

In this section, we will present experiments comparing our novel data structure based on sequential arrays with the approach we adapted based on~\cite{paper1} using \Btrees{} as a replacement for self-balanced binary search trees (BSTs).
Briefly, the approach introduced in~\cite{paper1} stores, in a matrix $n \times n$, pointers to BSTs containing time intervals.
In each BST, only non-redundant intervals are kept, \textit{i.e.}, those that do not contain another interval in the same tree.
The authors proposed to use \emph{join-based} operations in order to remove sequences of non-redundant intervals in $\log{\tau}$ time.
These operations can be found in~\ref{btreejoinsplit} for \Btrees{}.

In the following, we present two experiments in Sections~\ref{subsec:experiment1} and~\ref{subsec:experiment2}.
In the first one, we inserted unsorted contacts from complete temporal graphs, incrementally, in both data structures using the operation \alga{}.
In the second one, we inserted shuffled contacts from real-world datasets.

\subsection{Experiments with Synthetic Data}\label{subsec:experiment1}

In this first experiment, we generated complete temporal graphs with the number of vertices fixed to $100$ and varied the number of timestamps $\tau$ from $10$ to $10000$.
Then, we inserted their shuffled contacts in both data structures using the \alga{} operation.
The time to preallocate and initialize the arrays on disk for our array-based data structure was not considered in the total time.
We note that this extra cost can be high for large parameters; therefore, one should consider it whenever applicable.

\begin{figure}
  \centering
  \includegraphics[page=1, width=\textwidth]{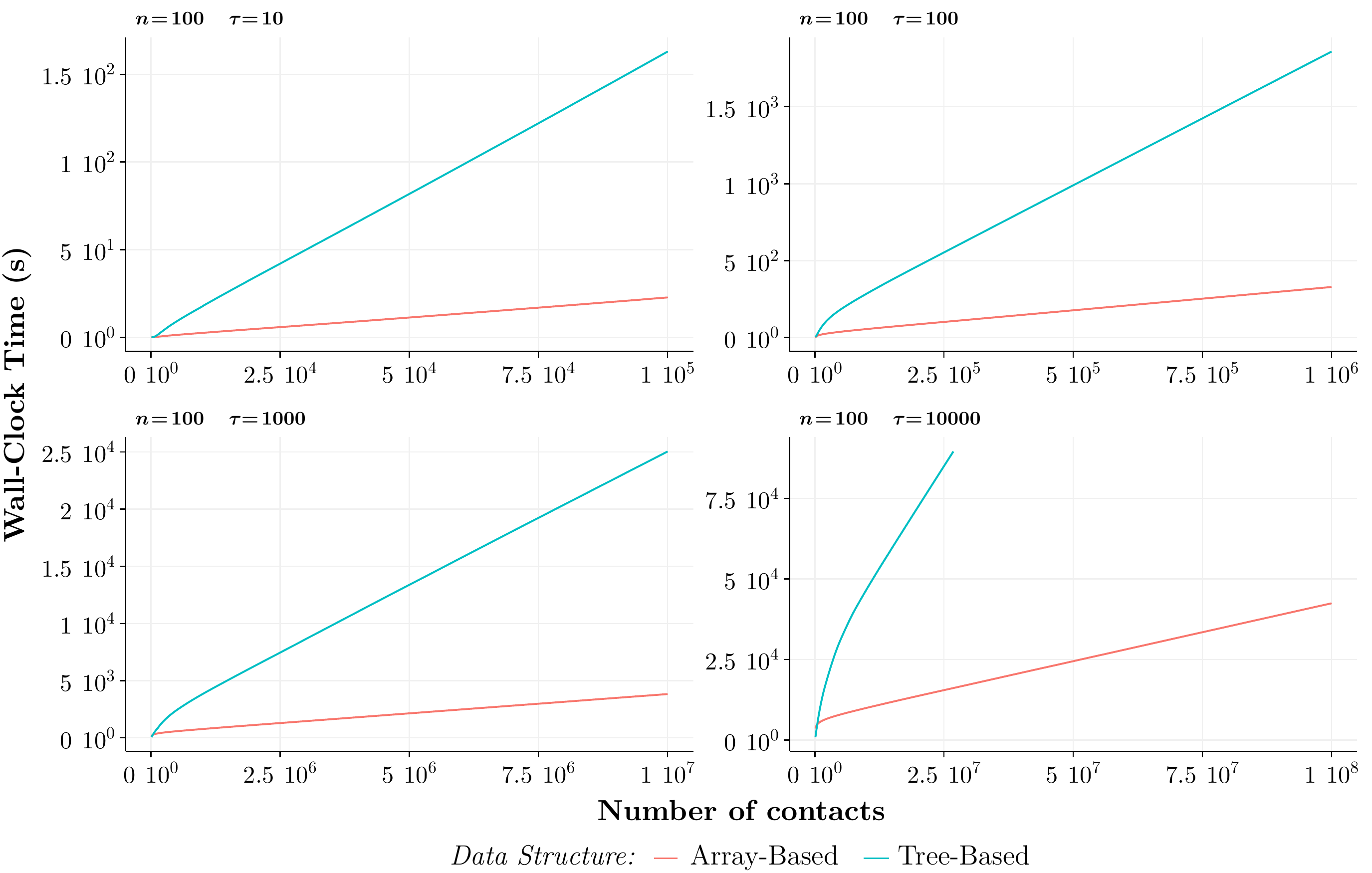}
  \caption{Cumulative wall-clock time to maintain data structures for reachability queries on synthetic data. We inserted shuffled contacts from complete temporal graphs into the data structures varying the number of timestamps $\tau$ while fixing the number of vertices to $100$. Red lines represent our novel data structure based on sequential arrays. Blue lines represent our adaptation of the approach introduced in~\cite{paper1} using \Btrees{} as self-balanced BSTs.}\label{fig:result1}
\end{figure}

Figure~\ref{fig:result1} shows the mean cumulative wall-clock time, averaged over $10$ executions, to maintain both data structures as new unsorted contacts were inserted.
We see that our novel data structure performs better for all configurations.
Even though the worst-case complexity of our algorithm for the \alga{} operation is linear in $\tau$ instead of logarithmic, it runs much faster using synthetic data.
We attribute this behavior to the fact that as new contacts are inserted, the probability of composing better \Rtuples, \textit{i.e.} journeys, decreases rapidly and, thus, our data structure updates on average only few cells per contact insertion.

Next, we will argue why the run time of our algorithm reduces with the addition of contacts.
Each pair of vertices $(u, v)$ are associated to a set $\mathcal{I}$ containing intervals $[t^-, t^+] \subseteq [1, \tau]$ in which $u$ can reach $v$ departing at $t^-$ and arriving at $t^+$.
For a particular pair of vertices, when an algorithm inserts a new interval $I$, all intervals $I'$ such that $I \subseteq I'$ can be safely removed since they become redundant.
Our data structure organizes these intervals in the arrays $M_{out}$ and $M_{in}$, which fix, respectively, the left and right endpoints, and our update algorithm discards redundant intervals by updating their cells accordingly by using Definition~\ref{def:redundancy}.

Consider the hierarchy of intervals illustrated in Figure~\ref{fig:possibility-tree}(a) for $\tau = 4$.
Each interval with length $l$ is linked to the intervals with length $l-1$ that it totally encloses.
For example, interval $[0, 5]$, with length $5$, links to intervals $[0, 4]$ and $[1, 5]$, with length $4$, because $[0, 4] \subseteq [0, 5]$ and $[1, 5] \subseteq [0, 5]$.
Initially, all intervals are available for insertion in our data structure.
When a new interval $[1, 2]$ is inserted, as show in Figure~\ref{fig:possibility-tree}(b), all intervals that contain it, including itself, are not available for insertion anymore.

Our update algorithm conceptually removes these intervals by drawing left and right frontiers separating available and non-available intervals starting from $[1, 2]$.
For instance, intervals $[1, 2]$ and $[0, 2]$, which belong to the left frontier, are updated in $M_{in}$ since they share the same right endpoint, and intervals $[1, 2]$, $[1, 3]$, $[1, 4]$ and $[1, 5]$, which belong to the right frontier, are updated in $M_{out}$ since they share the same left endpoint.
In this proccess, up to $\tau$ cells are updated in both $M_{in}$ and $M_{out}$.

Next, when a new interval $[3, 5]$ is inserted, as shown in Figure~\ref{fig:possibility-tree}(c), our algorithm must, again, draw the left and right frontiers starting from $[3, 5]$; however it does not need to advance previously drawn frontiers.
In this case, only intervals $[3, 5]$ and $[2, 5]$, which belong to the left frontier, are updated in $M_{in}$.
In Figure~\ref{fig:possibility-tree}(d), interval $[2, 3]$ is inserted and the same proccess repeats.
We see that as new intervals are inserted, the number of available intervals rapidly reduces.
Thus, even though our algorithm has complexity $O\left(\nicefrac{n^2\tau}{B}\right)$, it can run much faster when considering a sequence of contact insertions since the number of cells to be updated reduces over time.

Moveover, it is guaranteed that, for each new contact $(u, v, t)$, our algorithm will make unavailable for insertion every interval that is still available inside and at the frontiers starting from $[t, t + \delta]$ in the lowest level of the hierarchy associated with $(u, v)$.

\begin{figure}
  \centering
  \begin{tabular}{cc}
    \includegraphics[page=5, width=0.4\textwidth]{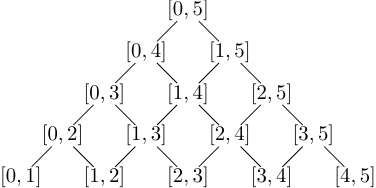} &
    \includegraphics[page=6, width=0.4\textwidth]{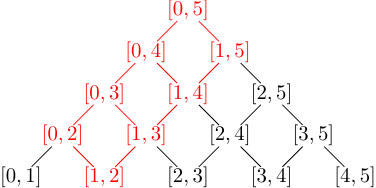} \\ [2mm]
    (a) all available intervals & (b) insert $[1, 2]$ \\ [5mm]
    \includegraphics[page=7, width=0.4\textwidth]{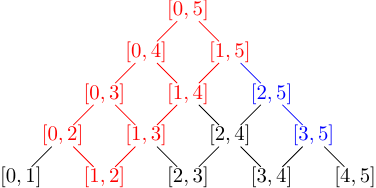} &
    \includegraphics[page=8, width=0.4\textwidth]{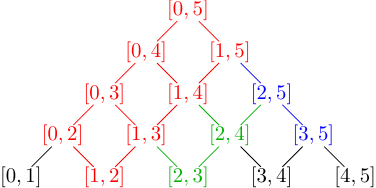} \\ [2mm]
    (c) insert $[3,5]$  & (d) insert $[2,3]$ \\
  \end{tabular}
  \caption{Illustration of the proccess performed by our update algorithm considering a fixed pair of vertices $(u, v)$ from a temporal graph with $\tau = 4$. Available intervals for insertion are colored in black, and invalidated intervals, \textit{i.e.} intervals that should not be considered anymore by our update algorithm, are colored in different colors.
  Links represent the direct containement relation between intervals with length $l$ and intervals with length $l-1$.}\label{fig:possibility-tree}
\end{figure}

\subsection{Experiments with Real-World Datasets}\label{subsec:experiment2}

In this second experiment, we downloaded small and medium real-world available on \url{https://networkrepository.com/dynamic.php}, and preprocessed them using our script available on \url{https://bitbucket.com/luizufu/temporalgraph-datasets-preprocessing}.
During the preprocessing, we rellabelled the vertices and shifted the timestamps of each dataset so that vertex identifiers were beetween $[1, n]$ and timestamp values start from $1$.
Then, we inserted the shuffled contacts of each dataset in both data structures using the \alga{} operation.
We assumed that all used datasets represent temporal digraphs, and we used $\delta = 1$, \textit{i.e.}, traversing any contact takes one time unit.


\begin{table}[ht]
\centering
\begin{tabular}{lrrrrll}
  \toprule
  dataset & $n$ & $\tau$ & contacts & density & \textbf{Array-Based} & \textbf{Tree-Based} \\
  \midrule
  aves-sparrow &  $52$ &   $2$ & $516$ & $0.1$ & \textcolor{blue}{$0.01 \pm 0$} & $0.07 \pm 0$ \\
  aves-weaver & $445$ &  $23$ & $1423$ & $0.003$ & \textcolor{blue}{$0.19 \pm 0$} & $1.16 \pm 0.01$ \\
  aves-wildbird & $202$ &   $6$ & $11900$ & $0.05$ & \textcolor{blue}{$0.97 \pm 0.01$} & $9.52 \pm 0.15$ \\
  ant-colony1 & $113$ &  $41$ & $111578$ & $0.46$ & \textcolor{blue}{$25.84 \pm 0.19$} & $161.3 \pm 1.03$ \\
  ant-colony2 & $131$ &  $41$ & $139925$ & $0.2$ & \textcolor{blue}{$43.96 \pm 0.49$} & $261.98 \pm 1.96$ \\
  ant-colony3 & $160$ &  $41$ & $241280$ & $0.23$ & \textcolor{blue}{$89.37 \pm 0.73$} & $524.79 \pm 6.71$ \\
  ant-colony4 & $102$ &  $41$ & $81599$ & $0.19$ & \textcolor{blue}{$16.49 \pm 0.12$} & $104.62 \pm 1.27$ \\
  ant-colony5 & $152$ &  $41$ & $194317$ & $0.21$ & \textcolor{blue}{$166.93 \pm 3.63$} & $526.03 \pm 144.32$ \\
  ant-colony6 & $164$ &  $39$ & $247214$ & $0.24$ & \textcolor{blue}{$88.99 \pm 0.93$} & $608.87 \pm 167.97$ \\
  copresence-LH10 &  $73$ & $259181$ & $150126$ & $0.0001$ & - & \textcolor{blue}{$61.5 \pm 0.53$} \\
  copresence-LyonSchool & $242$ & $117721$ & $6594492$ & $0.001$ & - & \textcolor{blue}{$14887.45 \pm 1576.49$} \\
  kilifi-within-households &  $54$ &  $59$ & $32426$ & $0.19$ & \textcolor{blue}{$0.09 \pm 0$} & $0.21 \pm 0$ \\
  mammalia-primate &  $25$ &  $19$ & $1340$ & $0.12$ & \textcolor{blue}{$0.09 \pm 0$} & $0.33 \pm 0.01$ \\
  mammalia-raccoon &  $24$ &  $52$ & $1997$ & $0.06$ & \textcolor{blue}{$0.21 \pm 0$} & $0.44 \pm 0.01$ \\
  mammalia-voles-bhp & $1686$ &  $63$ & $5324$ & $0.00003$ & \textcolor{blue}{$13.76 \pm 0.82$} & $19.22 \pm 0.24$ \\
  mammalia-voles-kcs & $1218$ &  $64$ & $4258$ & $0.00004$ & \textcolor{blue}{$7.92 \pm 0.27$} & $11.78 \pm 0.09$ \\
  mammalia-voles-plj & $1263$ &  $64$ & $3863$ & $0.00003$ & \textcolor{blue}{$6.18 \pm 0.23$} & $10.68 \pm 0.04$ \\
  mammalia-voles-rob & $1480$ &  $63$ & $4569$ & $0.00003$ & \textcolor{blue}{$10.28 \pm 0.41$} & $15.09 \pm 0.12$ \\
  tortoise-bsv & $136$ &   $4$ & $554$ & $0.008$ & \textcolor{blue}{$0.01 \pm 0$} & $0.14 \pm 0.01$ \\
  tortoise-cs &  $73$ &  $10$ & $258$ & $0.005$ & \textcolor{blue}{$0.01 \pm 0$} & $0.05 \pm 0$ \\
  tortoise-fi & $787$ &   $9$ & $1713$ & $0.0003$ & \textcolor{blue}{$0.15 \pm 0$} & $2.71 \pm 0.01$ \\
  trophallaxis-colony1 &  $41$ &   $8$ & $308$ & $0.02$ & \textcolor{blue}{$0.02 \pm 0$} & $0.06 \pm 0$ \\
  trophallaxis-colony2 &  $39$ &   $8$ & $330$ & $0.03$ & \textcolor{blue}{$0.02 \pm 0$} & $0.05 \pm 0$ \\
   \bottomrule
\end{tabular}
  \caption{Total wall-clock time in seconds to insert all shuffled contacts from real-world datasets with number of vertices $n$, number of timestamps $\tau$, number of contacts into data structures for reachability queries, and the density of the temporal graph represented by the dataset.
  Values were rounded to two decimal places.
  Array-based refers to our novel data structure and tree-based refers to our implementation of the approach introduced in~\cite{paper1} using \Btrees{} as BSTs replacement.
  Executions that reached the time limit of $5$ hours are marked with the symbol ``-''.
  }\label{tab:result2}
\end{table}

Table~\ref{fig:result1} shows the mean wall-clock time, averaged over $10$ executions, to insert all shuffled contacts of each dataset into both data structures.
We see that our novel data structure performs better on the majority of datasets.
However for the largest datasets, \texttt{copresence-LH10} and \texttt{copresence-LyonSchool},  the tree-based data structure performed better.
Both datasets have high values for $\tau$ and low density.
It means that, as density is too small, each insertion of a contact $(u, v, t)$ may trigger an initial update over arrays $M_{out}$ and $M_{in}$ that will touch many cells on disk.
As in Figure~\ref{fig:possibility-tree}(b), for most insertions, our update algorithm will draw left and right frontiers on the almost empty hierarchy associated with the pair of vertices $(u, v)$ starting from interval $[t, t + \delta]$.
Therefore, in this case, the linear factor on $\tau$ from the cost $O(\nicefrac{n^2\tau}{B})$ of our update algorithm will have a bigger impact on the run time since the sequence of insertions is not sufficiently long for our algorithm to benefit from later insertions.

\section{Concluding remarks}\label{sec:disk-conclusions}

We presented in this paper an incremental disk-based data structure to solve the dynamic connectivity problem in temporal graphs.
Our data structure prioritizes query time, answering reachability queries by accessing only one page.
Based on the ability to quickly retrieve reachability information among vertices inside time intervals, it can:
insert contacts in a non-chronological order accessing $O\left(\nicefrac{n^2\tau}{B}\right)$ pages, where $B$ is the size of disk pages;
check whether a temporal graph is connected within a time interval accessing $O\left(\nicefrac{n^2}{B}\right)$ pages,
and reconstruct journeys accessing $O\left(\nicefrac{n}{B}\right)$ pages.
Our algorithms exploit the special features of non-redundant (minimal) reachability information, which we represent explicitly through the concept of expanded \Rtuples.
As in~\cite{paper1}, the core of our data structure, is essentially a collection of non-redundant \Rtuples, whose size (and that of the data structure itself) cannot exceed $O\left(n^2\tau\right)$.
However, in our approach, all this space must be preallocated on disk.
The benefit of our data structure is that algorithms explicitly manage data sequentially and, therefore, it is more suitable for secondary memories in which random accesses are expensive.

Further investigations could be done in the direction of improving the complexity of our update algorithm.
Can \alga{} access less than $O\left(\nicefrac{n^2\tau}{B}\right)$ pages?
Another direction could be designing efficient disk-based data structures for the decremental and the fully-dynamic versions of this problem.
With \emph{unsorted} contact insertion and deletion, it seems to represent both a significant challenge and a natural extension of the present work, one that would certainly develop further our common understanding of temporal reachability.
Finally, it could be worth to investigate compressing algorithms to reduce the space of our data structure and the number of pages accessed by our update algorithm.
Specifically, we think that compression algorithms based on differences and run-length coding~\cite{compressionalgs} could achieve a very high compression rate since the arrays $M_{out}$ and $M_{in}$ store repeating ordered values.
The compressing schema could also solve the preallocation and initialization problem since all cells of $M_{out}$ and $M_{in}$ have, initially, the same value, which are very compressible.

\begin{appendices}

\section{Join and split operations for \Btrees}\label{sec:disk-experiments-preliminaries}

Let each leaf node of \Btrees{} contains an array $K$ of keys of size $N$ and a pointer to its next sibling.
Let each non-leaf node contains an array of keys $K$  of size $M$ and an additional array of pointers $C$ to child nodes of size $M + 1$.
Additionally, assume that every node contains the height of the sub-tree it belongs, a pointer to its leftmost leaf child and a pointer to its rightmost leaf child.
We note that we use these additional per-node data to simplify our algorithms and discussions.
In a real implementation, only the root node (the tree itself) must maintain them during the insertion and update operations.
Information regarding the rest of the nodes can be computed during the execution of the next algorithms without increasing complexities.

\subsection{\texttt{Join} operation on \Btrees}\label{subsec:join}

Algorithm~\ref{alg:join} performs the operation \texttt{join} for \Btrees.
Given two \Btrees{} $T_{left}$ and $T_{right}$, such that keys present in $T_{left}$ are smaller to keys in $T_{right}$, it must merge both trees in order to create a new valid \Btree{} $T$ containing all keys present in $T_{left}$ and $T_{right}$.
As \Btrees{} place leaf nodes at the same level, it simply inserts or shares the data present in the root node of the smaller tree into the appropriate node at the same height in the bigger tree.
Then it maintains the \Btree{} invariances up to its root node of the changed bigger tree whenever necessary.
First, in line~11, the algorithm sets the next sibling of the rightmost leaf of $T_{left}$ to be the leftmost leaf of $T_{right}$.
Then, if $\texttt{height}(T_{left}) \geq \texttt{height}(T_{right})$, in lines~13 and~14, it adds $T_{right}$ to $T_{left}$ by calling $\Call{joinRight}{T_{left}, T_{right}}$ and returns $T_{left}$; otherwise, in lines~16 and~17, it adds $T_{left}$ to $T_{right}$ by calling $\Call{joinLeft}{T_{left}, T_{right}}$ and returns $T_{right}$.
From lines~1 to~11, we detail the \textproc{joinRight} routine, the \textproc{joinLeft} routine is implemented symmetrically.
In line~1, the algorithm descends $T_{left}$ until reaching the rightmost node $n_{left}$ at the same height of the root node of $T_{right}$.
If a single node of size $B$ can fit the content of both $n_{left}$ and the root node of $T_{right}$, in line~5, it simply merges both nodes by adding to $n_{left}$ the data present in the root node of $T_{right}$.
Otherwise, in line~7, it equally shares the data of both nodes, and, in line~8, it inserts into the parent of $n_{lelft}$ a new key together with a pointer to the root node of $T_{right}$.
If the parent node has no space left to accommodate the new data, a node splitting routine must be invoked and this process can continue up to the root node of $T_{left}$.
Finally, if the algorithm needs to split the current root node of $T_{left}$, then it creates a new root node, and, in this case, In line~10, it increments the height of $T_{left}$ by one.

\begin{algorithm*}
  \caption{\texttt{join}}\label{alg:join}
  \begin{algorithmic}[1]
    \Require{Two trees of intervals $T_{left}$ and $T_{right}$}
    \Procedure{joinRight}{$T_{left}, T_{right}$}
      \State{$n_{left} \gets \texttt{descend\mund{}right}(T_{left}, \texttt{height}(T_{left}) - \texttt{height}(T_{right}))$}
      \State{$n_{right} \gets \texttt{root}(T_{right})$}
      \If{$\texttt{size}(n_{left}) + \texttt{size}(n_{right}) \leq B$}
        \State{$n_{left} \gets \texttt{merge}(n_{left}, n_{right})$}
      \Else{}
        \State{$\texttt{share}(n_{left}, n_{right})$}
        \State{$\texttt{insert\mund{}rec}(\texttt{parent}(n_{left}), \texttt{min\mund{}key}(n_{right}), n_{right})$}
        \If{new root node was created}
          \State{$\texttt{height}(T_{left}) \gets \texttt{height}(T_{left}) + 1$}
        \EndIf{}
      \EndIf{}
    \EndProcedure{}
    \vspace{0.2cm}
    \State{$\texttt{next\mund{}leaf}(\texttt{rightmost\mund{}leaf}(T_{left})) \gets \texttt{leftmost\mund{}leaf}(T_{right})$}
    \If{$\texttt{height}(T_{left}) \geq \texttt{height}(T_{right})$}
      \State{$\Call{joinRight}{T_{left}, T_{right}}$}
      \State{\textbf{return} $T_{left}$}
    \Else{}
      \State{$\Call{joinLeft}{T_{left}, T_{right}}$}
      \State{\textbf{return} $T_{right}$}
    \EndIf{}
  \end{algorithmic}
\end{algorithm*}

\begin{theorem}{Algorithm~\ref{alg:join} accesses $O(|\texttt{height}(T_{left}) - \texttt{height}(T_{right})|)$ pages in the worst-case.}
\end{theorem}

\begin{proof}
  In line~11, the algorithm accesses one page to set the next child node of the rightmost leaf node of $T_{left}$.
  Then, it calls \textproc{joinRight} or \textproc{joinLeft} depending on the heights of $T_{left}$ and $T_{right}$, both accessing the same amount of pages.
  Without loss of generality, assume that $\texttt{height}(T_{left}) \geq \texttt{height}(T_{right})$ and it calls thus the \textproc{joinRight} procedure.
  Then, at line~2, the algorithm accesses $\texttt{height}(T_{left}) - \texttt{height}(T_{right})$ pages while descending to the rightmost node of $T_{left}$ at height $\texttt{height}(T_{right})$.
  Next, if there is enough room to fit the data of both nodes being merged in a single node, it accesses $O(1)$ pages and the algorithm ends.
  Otherwise, it accesses $O(1)$ pages to share the content present in the considered nodes.
  Then, it accesses, again, $O(\texttt{height}(T_{left}) - \texttt{height}(T_{right}))$ pages in order to insert new key and pointer pairs up to the root of $T_{left}$ in the worst-case.
  Finally, if a new node is created, it access one more page to increment the height of $T_{left}$ and the algorithm ends.
  \qed
\end{proof}

\subsection{\texttt{Split} operation on \Btrees}\label{subsec:split}

Algorithm~\ref{alg:split} performs the operation \texttt{split} for \Btrees.
Given an interval key $L$, it must split a tree $T$ in two trees $T_{left}$ and $T_{right}$ such that all keys in $T_{left}$ are smaller than $I$ and all keys in $T_{right}$ are greater or equal to $I$.
To accomplish this task, it recursively descends $T$ from the root node to the leaf node containing the biggest key less than $L$ while partitioning nodes appropriately and, during the backward phase of the recursion, progressively building $T_{left}$ and $T_{right}$.
During each recursive step, in line~1, the algorithm first finds the position $k$ in the current root node such that $K[k] \geq I$, where $C[k]$ is the pointer that branches to the next child node for non-leaf nodes.
If the current root node is a leaf, in line~3, it partitions the current node in two sub-trees: $T_{left}$, containing a node with $K[1 \ldots k - 1]$; and $T_{right}$, containing a node with $K[k \ldots N]$.
Then, in line~4, it sets the next sibling of $T_{left}$'s root to $nil$; in line~5, it sets the next sibling of $T_{right}$ to the next sibling of $T$'s root; and, in line~6, it returns $(T_{left}, T_{right})$.
Note that no other leaf node besides the affected ones must update the pointer to its next siblings since the resulting trees will reuse the previous linkages.
Next, if the current node is a non-leaf, in line~7, the algorithm partitions the current node in three sub-trees: $T_{left}$, containing a root node with $K[1 \ldots k - 2]$ and $C[1 \ldots k - 1]$; $T_{child}$, containing a root node with $K[k - 1 \ldots k]$ and $C[k]$; and (3) $T_{right}$, containing a root node with $K[k + 1 \ldots M]$ and $C[k + 1 \ldots M + 1]$.
Additionally, whenever a sub-tree have only one pointer in its root node, its respective root node becomes the child pointed by it in order to maintain the correct \Btree{} layout.
Then, in line~8, it calls the \textproc{split} algorithm itself passing $T_{child}$ as parameter and obtaining two sub-trees $T'_{left}$ and $T'_{right}$ as the intermediate result.
Finally, in line~9, it advances the intermediate result by joining them appropriately with the sub-trees of the current level.

\begin{algorithm*}
  \caption{\texttt{split}}\label{alg:split}
  \begin{algorithmic}[1]
    \Require{A tree of intervals $T$ and a key interval L}

    \State{$k \gets \texttt{find\mund{}key\mund{}position}(\texttt{root}(T), L)$}
    \If{$\texttt{root}(T)$ is a leaf}
      \State{$(T_{left}, T_{right}) \gets \texttt{split\mund{}leaf}(root(T), k)$}
      \State{$\texttt{next\mund{}leaf}(\texttt{root}(T_{left})) \gets nil$}
      \State{$\texttt{next\mund{}leaf}(\texttt{root}(T_{right})) \gets \texttt{next\mund{}leaf}(\texttt{root}(T))$}
      \State{\textbf{return} $(T_{left}, T_{right})$}
    \EndIf{}
    \State{$(T_{left}, T_{child}, T_{right}) \gets \texttt{split\mund{}non\mund{}leaf}(root(T), k)$}
    \State{$(T'_{left}, T'_{right}) \gets \texttt{split}(T_{child}, L)$}
    \State{\textbf{return} $(\texttt{join}(T_{left}, T'_{left}), \texttt{join}(T'_{right}, T_{right}))$}
  \end{algorithmic}
\end{algorithm*}

\begin{theorem}{Algorithm~\ref{alg:split} accesses $O(\log_B{(\tau)})$ pages in the worst-case where $\tau$ is the maximum number of keys in the tree.}
\end{theorem}

\begin{proof}
  During each recursive step, Algorithm~\ref{alg:split} needs to: (1) partition the current node being considered in at least two sub-trees; (2) change pointers to next leaf siblings, whether the current node is a leaf; and (3) join the current sub-trees with the sub-trees resulting from the next recursive step.
  For (1), the algorithm accesses $O(1)$ pages.
  In the worst-case scenario, if the root node of a sub-tree has only a single child, the algorithm reads this child in order to make it the new root node.
  For (2), the algorithm also accesses $O(1)$ pages since only two leaf nodes are updated.
  For (3), the algorithm calls the \texttt{join} algorithm twice.
  Without loss of generality, consider only calls maintaining $T_{left}$.
  From the node above the leaf node containing the split key up to the root of $T$, the algorithm joins the current left sub-tree $T_{left}$ with the intermediate left sub-tree $T'_{left}$ resulting from the previous iteration.
  At each iteration there is a $\texttt{join}(T_{left}, T'_{left})$ call in which $T_{left}$ is either empty or non-empty.
  In case $T_{left}$ is empty, the algorithm pays nothing and, at the next iteration, the difference in height between $T_{left}$ and $T'_{left}$ increases by one.
  In case $T_{left}$ is non-empty, the algorithm pays the difference in height accumulated so far and, at the next iteration, the difference in height resets to one.
  Therefore, as the summation of all payments is at most the height of the tree, the algorithm accesses $O(\log_{B}(\tau))$ pages while processing all \texttt{join} calls.
  \qed
\end{proof}

\end{appendices}

\begin{paragraph}{Acknowledgements}
  This study was financed in part by Funda\c{c}\~{a}o de Amparo \`{a} Pesquisa do Estado de Minas Gerais (FAPEMIG) and the Coordena\c{c}\~{a}o de Aperfei\c{c}oamento de Pessoal de N\'{i}vel Superior - Brasil (CAPES) - Finance Code 001* - under the ``CAPES PrInt program'' awarded to the Computer Science Post-graduate Program of the Federal University of Uberl\^{a}ndia.
\end{paragraph}

\printbibliography{}

\end{document}